\documentclass{amsart}

\usepackage{amsmath, amsthm}
\usepackage{graphicx,psfrag,epsf}
\usepackage{enumerate}
\usepackage{natbib}
\usepackage{url}
\usepackage{booktabs}
\usepackage{graphicx}
\usepackage{tikz}
\usetikzlibrary{3d}

\newtheorem{theorem}{Theorem}[section]
\newtheorem{lemma}{Lemma}[section]
\newtheorem{corollary}{Corollary}[section]
\newtheorem{proposition}{Proposition}[section]
\theoremstyle{definition}
\newtheorem{definition}{Definition}[section]
\newtheorem{remark}{Remark}[section]
\allowdisplaybreaks
\numberwithin{table}{section}
\numberwithin{equation}{section}
\numberwithin{remark}{section}

\long\def\symbolfootnote[#1]#2{\begingroup%
 \def\thefootnote{\fnsymbol{footnote}}\footnote[#1]{#2}\endgroup}

\begin{document}

\title[Reversals of Least-Squares Estimates]{Reversals of Least-Squares Estimates and Model-Independent Estimation for Directions of Unique Effects}
\author{Brian Knaeble}
\address{Department of MSCS, University of Wisconsin - Stout}
\email{knaebleb@uwstout.edu}

\author{Seth Dutter}
\address{Department of MSCS, University of Wisconsin - Stout}
\email{dutters@uwstout.edu}

\date{March 4, 2015}


\begin{abstract}
When a linear model is adjusted to control for additional explanatory variables the sign of a fitted coefficient may reverse.  Here these reversals are studied using coefficients of determination.  The resulting theory can be used to determine directions of unique effects in the presence of substantial model uncertainty.  This process is called model-independent estimation when the estimates are invariant across changes to the model structure.  When a single covariate is added, the reversal region can be understood geometrically as an elliptical cone of two nappes with an axis of symmetry relating to a best-possible condition for a reversal using a single coefficient of determination.  When a set of covariates are added to a model with a single explanatory variable, model-independent estimation can be implemented using subject matter knowledge.  More general theory with partial coefficients is applicable to analysis of large data sets.  Applications are demonstrated with dietary health data from the United Nations.  Necessary conditions for Simpson's paradox are derived.
\end{abstract}

\maketitle

\section{Introduction}
A multivariate statistical model may be useful for predicting values of some variables from values of other variables for individuals throughout a population of study, yet the same model may be inaccurate when used to estimate the effects of experimental manipulation on the same individuals.  For example, standardized test scores of students can be predicted using information about school type, but effects of transferring students from one school to another may be difficult to ascertain.  In general, models may suggest effects that are confounded by a set of lurking variables.  \citet{Pearl09a} gives a causal definition for confounding in his book {\it Causality}, in contrast to definitions based on associational criteria used by ``epidemiologists, biostatisticians, social scientists, and economists.''  \citet{Greenland01} observe how in health research the term confounding has been used to refer to at least four distinct concepts---bias in estimating causal effects, noncollapsibility, inseperability of main effects and interactions, and 
inherent differences between variables measured and underlying constructs of interest.  Here we use the term confounding to refer to bias in estimating causal effects.  For further reading on confounding and related topics see \citet{Rosenbaum83a}, \citet{McNamee03}, and \citet{Howards12}.

Concerns about confounding lead to discussion of statistical adjustment.
\citet[Chapter 4]{Cox58} defines a concamitant variable through discussion of concamitant observations, which are supplementary observations (on a supplementary variable) that may be used to increase precision (of treatment effect estimates).  He describes how to adjust results for what ``would have been obtained had it been possible to make the concamitant variable the same for all (individuals).''   When fixed, the concamitant variable can not be responsible for observed variation in the outcome variable.  Adjustment is thus a way to mimic experimental control, and through adjustment researchers may say that they have controlled for a confounding variable.  Here we use the term covariate when referring to any variable that may be controlled for to facilitate adjustment.  
For more reading on adjustment methods that control for confounding see \citet{Lu09}.

Controlling for too many variables can be problematic \citep{Chatfield95,Hawkins04}, and controlling for certain types of variables can increase bias \citep{Robins86,Weinberg93,Scarborough10,Myers11,Pearl11}.  When subject matter specialists agree on the structure of a causal diagram it can be used to select an admissible set of covariates for adjustment \citep{Pearl09b}.  \citet{McNamee04} gives general advice for selecting a model, suggesting that both subject matter knowledge and statistical information should be used.  This was the approach taken by \citet{Davis12} during their study of the effect of rice consumption on (internal) exposure to arsenic in children.  They analyzed data, not only for ``rice consumption'' and ``urinary arsenic concentration'' (an indicator of recent exposure to arsenic), but also for ``age'', ``body mass index'', ``water source'', and other covariates.  They controlled for three different subsets of covariates by fitting three different multiple regression models, and within each model they interpreted the fitted coefficient for ``rice consumption'' as an adjusted estimate for the unique effect of rice consumption on arsenic exposure.  All three adjusted estimates were found to be statistically significant, yet the authors concluded that their study only ``suggests that rice consumption is a potential source of arsenic exposure.''  The authors displayed an awareness of what \citet{Chatfield95} calls model uncertainty.  For additional examples of regression in the presence of model uncertainty see \citet{Jungert12}, \citet{Nelson13}, \citet{Cervellati12}, and \citet{Lignell13}.  

We have discussed terminology and established context in order to state the central idea of this paper---if some aspect of an uncertain model is shown to be insensitive to adjustment by control for any subset of a larger set of covariates, and if all confounding variables are known to be within that larger set of covariates, then causal interpretation is more acceptable than otherwise.  In this way causal conclusions may be obtained with a combination of subject matter knowledge and sensitivity analysis.  We thus seek to develop useful mathematics that facilitates such sensitivity analysis.  We adopt the general context of linear regression, and we assume the principle of least squares.  Our objective is to identify simple conditions that can be used to ensure that estimates for directions of unique effects are invariant across many different model extensions.  The general process of using these conditions for the purpose of estimation is called {\it model-independent estimation}, and the mathematics associated with directions of effects is referred to as {\it analysis of reversals}.  The main results are presented in Section \ref{resultssec}.  Proofs are in Section \ref{proofsec}.  Necessary conditions for Simpson's paradox are derived in Section \ref{simpsec}.  Applications are demonstrated in Section \ref{miesec}.  Further discussion occurs in Section \ref{dsec}.

\section{Results}
\label{resultssec}

Let $\mathbf{y}$ denote a matrix with a single column of response data associated with the response variable $Y$.  Let $\mathbf{x}$ denote a matrix with a single column of explanatory data associated with the explanatory variable $X$.  Let $\mathbf{w}=[\mathbf{w}_1,...,\mathbf{w}_p]$ denote a matrix with $p$ columns of covariate data, associated with covariates $W_1,...,W_p$.  Let $\mathbf{u}=[\mathbf{u}_1,...,\mathbf{u}_k]$ denote a matrix with $k$ additional columns of covariate data, associated with covariates $U_1,...,U_k$.  Let $\mathbf{e}$ denote a matrix with a single column of ones.  All matrices have $n$ rows, and each row of $[\mathbf{y}\, \mathbf{x}\, \mathbf{w}\, \mathbf{u}]$ represents a multivariate observation on a single individual.  We refer to the columns of $[\mathbf{e}\, \mathbf{y}\, \mathbf{x}\, \mathbf{w}\, \mathbf{u}]$ as vectors and assume that each subset of vectors is linearly independent and non orthogonal.

The mathematics herein requires notation capable of representing coefficients across multiple models.  Let $\mathbf{m}=[\mathbf{m}_2\,...\,\mathbf{m}_l]$ denote a generic matrix with $l-1$ columns and $n$ rows, and let $\mathbf{m}_1$ and $\mathbf{z}$ denote generic vectors each with $n$ entries.
When $\mathbf{z}$ is regressed onto $[\mathbf{e}\,\mathbf{m}]$ we write $R^2(\mathbf{m},\mathbf{z})$ for the coefficient of determination and $R(\mathbf{m},\mathbf{z})$ for its positive square root.  For any $j$ we write $r(\mathbf{m}_j,\mathbf{z})$ for the correlation between $\mathbf{m}_j$ and $\mathbf{z}$.  We write $\mathbf{z}_{|\mathbf{m}}$ for the residual vector $\mathbf{z}-\hat{\mathbf{z}}(\mathbf{m})$, where $\hat{\mathbf{z}}(\mathbf{m})$ is the vector of fitted values.  A hat vector with a null argument is interpreted as the zero vector.  In place of $[{\mathbf{u}_1}_{|\mathbf{w}}\,...\,{\mathbf{u}_k}_{|\mathbf{w}}]$ we write $\mathbf{u}_{|\mathbf{w}}$.  When $\mathbf{z}$ is regressed onto $[\mathbf{e}\,\mathbf{m}_1\,\mathbf{m}]$ we write $\hat{\beta}_{\mathbf{m}_1|\mathbf{m}}(\mathbf{z})$ for the least squares fitted coefficient of $\mathbf{m}_1$.  Similar notation is used for other least-squares coefficients.  If $\mathbf{m}$ naturally decomposes column-wise we may then express $\mathbf{m}$ as a set of components separated by commas.
\begin{proposition}
\label{firstresult}
A reversal, \[\mathrm{sign}(\hat{\beta}_{\mathbf{x}|\mathbf{w},\mathbf{u}}(\mathbf{y}))\neq \mathrm{sign}(\hat{\beta}_{\mathbf{x}|\mathbf{w}}(\mathbf{y})),\]
occurs if and only if
\begin{equation}
\label{theeq}
\frac{R(\mathbf{u}_{|\mathbf{w}},\mathbf{x}_{|\mathbf{w}})R(\mathbf{u}_{|\mathbf{w}},\mathbf{y}_{|\mathbf{w}})r(\widehat{\mathbf{x}_{|\mathbf{w}}}(\mathbf{u}),\widehat{\mathbf{y}_{|\mathbf{w}}}(\mathbf{u}))}{r(\mathbf{x}_{|\mathbf{w}},\mathbf{y}_{|\mathbf{w}})}>1.
\end{equation}
\end{proposition}

We refer to $r(\mathbf{x}_{|\mathbf{w}},\mathbf{y}_{|\mathbf{w}})$ as the partial correlation between $\mathbf{x}$ and $\mathbf{y}$ given $\mathbf{w}$, and we denote it with $r_{\mathbf{x},\mathbf{y}|\mathbf{w}}$.  Likewise, we write $R_{\mathbf{u},\mathbf{x}|\mathbf{w}}$ for $R(\mathbf{u}_{|\mathbf{w}},\mathbf{x}_{|\mathbf{w}})$ and $R_{\mathbf{u},\mathbf{y}|\mathbf{w}}$ for $R(\mathbf{u}_{|\mathbf{w}},\mathbf{y}_{|\mathbf{w}})$. Since $|r|<1$ and additional explanatory columns can not decrease $R$, we have the following corollaries.
\begin{corollary}
\label{oneway1}
Let $\mathbf{s}$ be any subset of $\{\mathbf{u}_1,...,\mathbf{u}_k\}$.  Then
\[R_{\mathbf{u},\mathbf{x}|\mathbf{w}}R_{\mathbf{u},\mathbf{y}|\mathbf{w}}<|r_{\mathbf{x},\mathbf{y}|\mathbf{w}}| \implies \mathrm{sign}(\hat{\beta}_{\mathbf{x}|\mathbf{w},\mathbf{s}}(\mathbf{y}))= \mathrm{sign}(\hat{\beta}_{\mathbf{x}|\mathbf{w}}(\mathbf{y})).\]
\end{corollary}
\begin{definition}
\label{vvv}
\[\mathbf{v}(\mathbf{x},\mathbf{y};\mathbf{w})=\frac{\mathbf{x}_{|\mathbf{w}}}{|\mathbf{x}_{|\mathbf{w}}|}+\frac{\mathbf{y}_{|\mathbf{w}}}{|\mathbf{y}_{|\mathbf{w}}|}\]
\end{definition}
\begin{definition}
\[r^*=|2r_{\mathbf{x},\mathbf{y}|\mathbf{w}}/(r_{\mathbf{x},\mathbf{y}|\mathbf{w}}+1)|\]
\end{definition}
\begin{corollary}
\label{oneway2}
Let $\mathbf{s}$ be any subset of $\{\mathbf{u}_1,...,\mathbf{u}_k\}$.  Then
\[R^2(\mathbf{u}_{|\mathbf{w}},\mathbf{v})<r^* \implies \mathrm{sign}(\hat{\beta}_{\mathbf{x}|\mathbf{w},\mathbf{s}}(\mathbf{y}))= \mathrm{sign}(\hat{\beta}_{\mathbf{x}|\mathbf{w}}(\mathbf{y})).\]
\end{corollary}

The conclusions of both corollaries are identical.  Proposition \ref{firstresult} is logically stronger than Corollary \ref{oneway1}, and Corollary \ref{oneway1} is logically stronger than Corollary \ref{oneway2}.  The condition within Corollary \ref{oneway2} is best possible (see Remark \ref{bestpos}) based on a single coefficient of determination for our desired conclusion.  The conclusion makes model-independent estimation possible for the direction of an effect.  All $2^k$ subsets of $\{\mathbf{u}_1,...,\mathbf{u}_k\}$ are handled simultaneously.  When $k$ is large, model-independent estimation can complement Bayesian model averaging (see \citet{Hoeting99}).  When $\mathbf{w}=\emptyset$, model-independent estimation can be implemented using only subject matter knowledge regarding hypothetical $\mathbf{u}$.  This is because coefficients of determination are intuitive.  Intuition for reversals of least-squares estimates and intuition relating to general adjustment of regression models can be improved through study of the imagery in Figure \ref{cols}.  When $\mathbf{u}$ refers to a single vector then $|r(\widehat{\mathbf{x}_{|\mathbf{w}}}(\mathbf{u}),\widehat{\mathbf{y}_{|\mathbf{w}}}(\mathbf{u}))|=1$, and the condition $R_{\mathbf{u},\mathbf{x}|\mathbf{w}}R_{\mathbf{u},\mathbf{y}|\mathbf{w}}>|r_{\mathbf{x},\mathbf{y}|\mathbf{w}}|$ is necessary and sufficient for a reversal.  Within the space orthogonal to the columns of $[\mathbf{e}\,\mathbf{w}]$ the reversal region for $\mathbf{u}_{|\mathbf{w}}$ is an ellipsoidal cone of two nappes (see (\ref{elcone})), with axis of symmetry along $\mathbf{v}$ and boundary vectors having coefficients of determination greater than or equal to $r^*$.

\begin{figure}[b]
\centering
\begin{tikzpicture}[scale=2.7]

\begin{scope}[canvas is zy plane at x=0]
\draw[line width=.6pt,->,color=purple,dashed] (0,0) -- (-1.5,3);
\draw[line width=2.0pt,->,color=purple] (0,0) -- (1.5,3);
\draw[line width=2.0pt,->,color=purple] (-1.08,2.16) -- (-1.5,3);
\end{scope}

\begin{scope}[canvas is xy plane at z=0]
\draw[line width=.6pt,->,color=red,dashed] (0,0) -- (-1.34,3);
\draw[line width=.6pt,->,color=red,dashed] (0,0) -- (1.34,3);
\draw[line width=.6pt, dashed, ->] (0,0) -- (0,3);
\draw[line width=1.0pt,color=blue] (0,0) -- (-1.581,2.70);
\draw[line width=.6pt,color=red,dashed] (0,0) -- (-1.445,2.7);
\draw[line width=1.0pt,color=blue] (0,0) -- (1.5,3);

\draw[line width=1.0pt,color=red]  (-1.36408,2.5488) -- (-1.445,2.7);

\draw[line width=2.0pt,->,color=red] (-1.12694,2.523) -- (-1.34,3);
\draw[line width=2.0pt,->] (0,2.48) -- (0,3);
\draw[line width=2.0pt,->,color=red] (1.273,2.85) -- (1.34,3);
\draw (-1.09,2.98) node {$\mathbf{x}_{|\mathbf{w}}$};
\draw (.13,3.00) node {$\mathbf{v}$};
\draw (1.14,3.01) node {$\mathbf{y}_{|\mathbf{w}}$};
\end{scope}

\begin{scope}[canvas is zx plane at y=3]
\draw[color=red, line width=1.2pt] ((0,0) ellipse (1.5cm and 1.34cm);
\draw[color=blue, line width=1.2pt] ((0,0) ellipse (1.5cm and 1.5cm);
\end{scope}

\end{tikzpicture}
\caption{A vector $\mathbf{u}$ has induced a reversal, $\mathrm{sign}(\hat{\beta}_{\mathbf{x}|\mathbf{w},\mathbf{u}}(\mathbf{y}))\neq \mathrm{sign}(\hat{\beta}_{\mathbf{x}|\mathbf{w}}(\mathbf{y}))$, if and only if within the span of $\{\mathbf{x}_{|\mathbf{w}},\mathbf{y}_{|\mathbf{w}},\mathbf{u}_{|\mathbf{w}}\}$ we have $\mathbf{u}_{|\mathbf{w}}$ or $-\mathbf{u}_{|\mathbf{w}}$ positioned within the red, elliptical cone.  The blue, spherical cone relates to Corollary $\ref{oneway2}$, and the square of the correlation between either purple vector and $\mathbf{v}$ is $r^*$.}
\label{cols}
\end{figure}
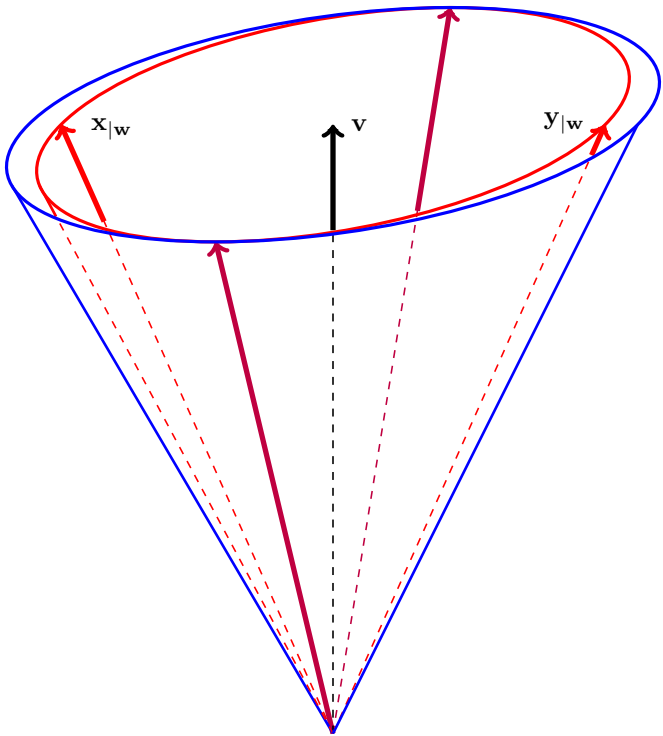
\section{Proofs}
\label{proofsec}
Corollary \ref{oneway1} is a ready consequence of Proposition \ref{firstresult}.  It thus remains to prove Proposition \ref{firstresult} and Corollary \ref{oneway2}.  Note how with $\perp$ indicating orthogonality between sets of vectors we have
\begin{equation}\label{takeout1}\{\mathbf{e},\mathbf{w}_1,...,\mathbf{w}_p\}\perp \{\mathbf{x}_{|\mathbf{w}},\mathbf{y}_{|\mathbf{w}},{\mathbf{u}_1}_{|\mathbf{w}},...,{\mathbf{u}_k}_{|\mathbf{w}}\},\end{equation}
and therefore
\begin{equation}\label{takeout}\hat{\beta}_{\mathbf{x}|\mathbf{w},\mathbf{u}}(\mathbf{y})=\hat{\beta}_{\mathbf{x}_{|\mathbf{w}}|\mathbf{u}_{|\mathbf{w}}}(\mathbf{y}_{|\mathbf{w}})~\text{and}~\hat{\beta}_{\mathbf{x}|\mathbf{w}}(\mathbf{y})=\hat{\beta}_{\mathbf{x}_{|\mathbf{w}}}(\mathbf{y}_{|\mathbf{w}}).\end{equation}
Let $\vec{x}$ stand for $\mathbf{x}_{|\mathbf{w}}$, $\vec{y}$ stand for $\mathbf{y}_{|\mathbf{w}}$, and $\vec{u}=[\vec{u}_1\,\cdots\,\vec{u}_k]$ stand for $\mathbf{u}_{|\mathbf{w}}$.  To prove Proposition \ref{firstresult} it thus suffices to demonstrate
\[\mathrm{sign}(\hat{\beta}_{\vec{x}|\vec{u}}(\vec{y}))\neq \mathrm{sign}(\hat{\beta}_{\vec{x}}(\vec{y}))\iff
\frac{R(\vec{u},\vec{x})R(\vec{u},\vec{y})r(\hat{\vec{x}}(\vec{u}),\hat{\vec{y}}(\vec{u}))}{r(\vec{x},\vec{y})}>1.\]

Since $\{\vec{x},\vec{y},\vec{u}_1,...,\vec{u}_k\}\perp \{\mathbf{e}\}$, each element of the set $\{\vec{x},\vec{y},\vec{u}_1,...,\vec{u}_k\}$ is a centered (mean zero) vector.  We can assume also that each vector is unit length and that $\{\vec{u}_1,...,\vec{u}_k\}$ is an orthonormal subset.
When $\vec{y}$ is regressed onto $[\vec{x}\,\vec{u}]$ the vector of fitted coefficients is \[\mathbf{\beta}= [\hat{\beta}_{\vec{x}|\vec{u}}(\vec{y})\,\hat{\beta}_{\vec{u}_1|\vec{x},\vec{u}_2,...,\vec{u}_k}(\vec{y})\,\cdots\,\hat{\beta}_{\vec{u}_k|\vec{x},\vec{u}_1,...,\vec{u}_{k-1}}(\vec{y})]^t.\]  With ${A} = [\vec{x}\,\vec{u}_1\,\cdots\,\vec{u}_k]$ the normal equations are
\[
\left({A}^t{A}\right)\hat{\beta}= {A}^t\vec{y}.
\]
Set ${B} = [\vec{y}\, \vec{u}_1\,\cdots\,\vec{u}_k]$.  Replacing the first column of ${A}^t{A}$ with ${A}^t\vec{y}$ produces the matrix ${A}^t{B}$, and by Cramer's rule
\begin{equation}
\label{dets}
\hat{\beta}_{\vec{x}|\vec{u}}(\vec{y}) = \frac{\det({A}^t{B})}{\det({A}^t{A})}.
\end{equation}

Because $\vec{u}$ is orthonormal,
\[
{A}^t{B} = \left[
                           \begin{array}{cccc}
                             \langle\vec{x},\vec{y}\rangle & \langle\vec{x},\vec{u}_1\rangle & \cdots & \langle \vec{x},\vec{u}_k\rangle \\
                             \langle\vec{u}_1,\vec{y}\rangle & 1 & \cdots & 0 \\
                             \vdots & \vdots & \ddots & \vdots \\
                             \langle\vec{u}_k,\vec{y}\rangle & 0 & \cdots & 1 \\
                           \end{array}
                         \right]
\]
and
\[
{A}^t{A} = \left[
                           \begin{array}{cccc}
                             \langle\vec{x},\vec{x}\rangle & \langle\vec{x},\vec{u}_1\rangle & \cdots & \langle\vec{x},\vec{u}_k\rangle \\
                             \langle\vec{u}_1,\vec{x}\rangle & 1 & \cdots & 0 \\
                             \vdots & \vdots & \ddots & \vdots \\
                             \langle\vec{u}_k,\vec{x}\rangle & 0 & \cdots & 1 \\
                           \end{array}
                         \right].
\]
The determinants from \ref{dets} can thus be evaluated using the Leibniz formula, and the result is
\begin{equation}
\label{closer}
\hat{\beta}_{\vec{x}|\vec{u}}(\vec{y})=\frac{\langle \vec{x},\vec{y} \rangle -\sum_{i=1}^{k}\langle \vec{x},\vec{u}_i\rangle\langle\vec{u}_i,\vec{y}\rangle}{\langle \vec{x},\vec{x} \rangle -\sum_{i=1}^{k}\langle \vec{x},\vec{u}_i\rangle\langle\vec{u}_i,\vec{x}\rangle}=\frac{\langle \vec{x},\vec{y}\rangle-\langle\hat{\vec{x}}(\vec{u}),\hat{\vec{y}}(\vec{u})\rangle}{\langle \vec{x},\vec{x}\rangle-\langle\hat{\vec{x}}(\vec{u}),\hat{\vec{x}}(\vec{u})\rangle}.
\end{equation}

With centered, unit-length data, we have $\hat{\beta}_{\vec{x}}(\vec{y})=\langle \vec{x},\vec{y} \rangle = r(\vec{x},\vec{y})$, $R(\vec{u},\vec{x})=|\hat{\vec{x}}(\vec{u})|$, and $R(\vec{u},\vec{y})=|\hat{\vec{y}}(\vec{u})|$.  These observations allow us to manipulate (\ref{closer}).  After multiplying by $\alpha:=\left(\langle \vec{x},\vec{x}\rangle-\langle\hat{\vec{x}}(\vec{u}),\hat{\vec{x}}(\vec{u})\rangle\right)$ the result is
\begin{align}
\label{line1} \alpha \hat{\beta}_{\vec{x}|\vec{u}}(\vec{y})&=\langle \vec{x},\vec{y}\rangle-\langle\hat{\vec{x}}(\vec{u}),\hat{\vec{y}}(\vec{u})\rangle \\ \nonumber
\frac{\alpha\hat{\beta}_{\vec{x}|\vec{u}}(\vec{y})}{\hat{\beta}_{\vec{x}}(\vec{y})}&=1-\frac{\langle\hat{\vec{x}}(\vec{u}),\hat{\vec{y}}(\vec{u})\rangle}{r(\vec{x},\vec{y})} \\
\nonumber
\alpha\frac{\hat{\beta}_{\vec{x}|\vec{u}}(\vec{y})}{\hat{\beta}_{\vec{x}}(\vec{y})}&=1-\frac{|\hat{\vec{x}}(\vec{u})||\hat{\vec{y}}(\vec{u})|}{r(\vec{x},\vec{y})}\frac{\langle\hat{\vec{x}}(\vec{u}),\hat{\vec{y}}(\vec{u})\rangle}{|\hat{\vec{x}}(\vec{u})||\hat{\vec{y}}(\vec{u})|}\\ \label{almost}
\alpha\frac{\hat{\beta}_{\vec{x}|\vec{u}}(\vec{y})}{\hat{\beta}_{\vec{x}}(\vec{y})}&=1-\frac{R(\vec{u},\vec{x})R(\vec{u},\vec{y})r(\hat{\vec{x}}(\vec{u}),\hat{\vec{y}}(\vec{u}))}{r(\vec{x},\vec{y})}.
\end{align}
Because $\alpha>0$, we see from (\ref{almost}) that \[\mathrm{sign}(\hat{\beta}_{\vec{x}|\vec{u}}(\vec{y}))\neq \mathrm{sign}(\hat{\beta}_{\vec{x}}(\vec{y}))\iff
\frac{R(\vec{u},\vec{x})R(\vec{u},\vec{y})r(\hat{\vec{x}}(\vec{u}),\hat{\vec{y}}(\vec{u}))}{r(\vec{x},\vec{y})}>1.\]
This completes the proof of Proposition \ref{firstresult}.

To demonstrate the truth of Corollary \ref{oneway2} we remain in the same context.  Each of the vectors in the set $\{\vec{x},\vec{y},\vec{u}_1,...,\vec{u}_k\}$ is centered and unit length, and we now additionally consider $\vec{v}=\vec{x}+\vec{y}$.  Note that $\vec{v}$ is equal to $\mathbf{v}$ from Definition \ref{vvv}.  We assume $\hat{\beta}_{\vec{x}}(\vec{y})>0$.  This can be assumed without loss of generality by replacing $\vec{x}$ with $-\vec{x}$ if necessary.  We show \begin{equation}\label{imp}R^2(\vec{u},\vec{v})<2r(\vec{x},\vec{y})/(1+r(\vec{x},\vec{y}))\implies \hat{\beta}_{\vec{x}|\vec{u}}(\vec{y})>0.\end{equation}
The condition for the implication within (\ref{imp}) can be written as
\begin{align}
\nonumber 2\langle \vec{x},\vec{y}\rangle/(1+\langle\vec{x},\vec{y}\rangle)&>|\hat{\vec{v}}(\vec{u})|^2/|\vec{v}|^2\\
\nonumber 2\langle \vec{x},\vec{y}\rangle/(1+\langle\vec{x},\vec{y}\rangle)&>|\hat{\vec{v}}(\vec{u})|^2/(2(1+\langle\vec{x},\vec{y}\rangle))\\
\nonumber 4\langle \vec{x},\vec{y}\rangle&>|\hat{\vec{v}}(\vec{u})|^2\\
\nonumber |\vec{v}|^2+2(\langle\vec{x},\vec{y}\rangle -1) &> |\hat{\vec{v}}(\vec{u})|^2\\
\nonumber \frac{1}{2}\left(|\vec{v}|^2-|\hat{\vec{v}}(\vec{u})|^2\right)+\langle\vec{x},\vec{y}\rangle -1&>0\\
\frac{1}{2}\left(|\vec{v}|^2-|\hat{\vec{v}}(\vec{u})|^2\right)+\langle\vec{x},\vec{y}\rangle-(\langle\vec{x},\vec{y}\rangle-\langle \hat{\vec{x}}(\vec{u}),\label{lastone}\hat{\vec{y}}(\vec{u})\rangle)-1&>-(\langle\vec{x},\vec{y}\rangle-\langle \hat{\vec{x}}(\vec{u}),\hat{\vec{y}}(\vec{u})\rangle).
\end{align}
Since $|\hat{\vec{x}}(\vec{u})||\hat{\vec{y}}(\vec{u})|\geq \langle\hat{\vec{x}}(\vec{u}),\hat{\vec{y}}(\vec{u})\rangle=\langle\vec{x},\vec{y}\rangle-(\langle\vec{x},\vec{y}\rangle-\langle \hat{\vec{x}}(\vec{u}),\hat{\vec{y}}(\vec{u})\rangle)$ we have
\[\frac{1}{2}\left(|\vec{v}|^2-|\hat{\vec{v}}(\vec{u})|^2\right)+|\hat{\vec{x}}(\vec{u})||\hat{\vec{y}}(\vec{u})|-1>-(\langle\vec{x},\vec{y}\rangle-\langle \hat{\vec{x}}(\vec{u}),\hat{\vec{y}}(\vec{u})\rangle),\]
and via Jensen's inequality $\frac{1}{2}(|\hat{\vec{x}}(\vec{u})|^2+|\hat{\vec{y}}(\vec{u})|^2)\geq\frac{1}{4}(|\hat{\vec{x}}(\vec{u})|+|\hat{\vec{y}}(\vec{u})|)^2\geq |\hat{\vec{x}}(\vec{u})||\hat{\vec{y}}(\vec{u})|$.  Therefore,
\begin{equation}
\label{nextone2}
\frac{1}{2}\left(|\vec{v}|^2-|\hat{\vec{v}}(\vec{u})|^2\right)+\frac{1}{2}(|\hat{\vec{x}}(\vec{u})|^2+|\hat{\vec{y}}(\vec{u})|^2)-1>-(\langle\vec{x},\vec{y}\rangle-\langle \hat{\vec{x}}(\vec{u}),\hat{\vec{y}}(\vec{u})\rangle).
\end{equation}
 \begin{remark}
\label{bestpos}
(\ref{lastone}) and (\ref{nextone2}) are logically equivalent if and only if $\hat{\vec{x}}(\vec{u})=\hat{\vec{y}}(\vec{u})$.
\end{remark}
\noindent Completing the square gives \begin{equation}\label{compsq}(\langle\vec{x},\vec{y}\rangle-\langle \hat{\vec{x}}(\vec{u}),\hat{\vec{y}}(\vec{u})\rangle)=\frac{1}{2}\left(|\vec{v}|^2-|\hat{\vec{v}}(\vec{u})|^2\right)+\frac{1}{2}(|\hat{\vec{x}}(\vec{u})|^2+|\hat{\vec{y}}(\vec{u})|^2)-1.\end{equation}
Substitution of $(\langle\vec{x},\vec{y}\rangle-\langle \hat{\vec{x}}(\vec{u}),\hat{\vec{y}}(\vec{u})\rangle)$ for $\frac{1}{2}\left(|\vec{v}|^2-|\hat{\vec{v}}(\vec{u})|^2\right)+\frac{1}{2}(|\hat{\vec{x}}(\vec{u})|^2+|\hat{\vec{y}}(\vec{u})|^2)-1$ within (\ref{compsq}) thus leads to
\[\langle\vec{x},\vec{y}\rangle-\langle \hat{\vec{x}}(\vec{u}),\hat{\vec{y}}(\vec{u})\rangle>0,\]
which by (\ref{line1}) is the desired conclusion of (\ref{imp}).  This completes the proof of Corollary \ref{oneway2}.

By (\ref{takeout1}) and (\ref{takeout}), when $k=1$, the reversal region for $\mathbf{u}$ consists of those points within the column space that project onto a region, $V(\vec{x},\vec{y})$, within the space that is orthogonal to the columns of $[\mathbf{e}\,\mathbf{w}]$.  To see how $V$ is an ellipsoidal cone, set $r=|r(\vec{x},\vec{y})|$, scale $\vec{x}$ and $\vec{y}$ (perhaps negatively), and select orthonormal coordinates for that space of $m=n-p-1$ dimensions so that
\[\vec{x}=\left(-\sqrt{\frac{1-r}{2}}, \sqrt{\frac{1+r}{2}},0,...,0\right)\mathrm{~and~}\vec{y}=\left(\sqrt{\frac{1-r}{2}}, \sqrt{\frac{1+r}{2}},0,...,0\right).\] Let $u=(u_1,...,u_m)$ be a variable vector in that same space.  We have
\[\hat{\vec{x}}(u)=\left(-u_1\sqrt{\frac{1-r}{2}}+u_2\sqrt{\frac{1+r}{2}}\right)\frac{u}{|u|^2}\mathrm{~and~}\hat{\vec{y}}(u)=\left(u_1\sqrt{\frac{1-r}{2}}+u_2\sqrt{\frac{1+r}{2}}\right)\frac{u}{|u|^2}.\]  Therefore, via (\ref{line1}), $\hat{\beta}_{\vec{x},\vec{y}|u}=0$ if and only if
\begin{align}
\nonumber 2r&=2\left(\langle \hat{\vec{x}}(u),\hat{\vec{y}}(u)\rangle\right)\\
\nonumber2r&=
 2\left(-u_1\sqrt{\frac{1-r}{2}}+u_2\sqrt{\frac{1+r}{2}}\right)\left(u_1\sqrt{\frac{1-r}{2}}+u_2\sqrt{\frac{1+r}{2}}\right)\frac{|u|^2}{|u|^4}\\
\nonumber 2r&=2\left(u_2^2\frac{(1+r)}{2}-u_1^2\frac{(1-r)}{2}\right)\frac{1}{|u|^2}\\
\label{you}2r|u|^2&=\left(u_2^2(1+r)-u_1^2(1-r)\right).
\end{align}
With $u_2=1$, line (\ref{you}) can be written as
\begin{align}
\nonumber 2r(u_1^2+1+u_3^2+...+u_m^2)&=(1+r)+(r-1)u_1^2\\
\nonumber (r+1)u_1^2+2r(u_3^2+...+u_m^2)&=1-r\\
\label{elcone}\frac{1+r}{1-r}u_1^2+\frac{2r}{1-r}(u_3^2+...+u_m^2)&=1.
\end{align}
Since scaling of $u$ does not affect $\hat{\beta}_{\vec{x},\vec{y}|u}$, the zero set $\{u:\hat{\beta}_{\vec{x},\vec{y}|u}=0\}$ is conical, of two nappes, with ellipsoidal cross-sections.
The cross sections are approximately spherical for large values of $r$.

\section{Applications}
Analysis of reversals has produced the mathematical results of Section \ref{resultssec}.  These results can be used when the direction of a unique effect is of interest, as could be the case during study of the safety of a medical intervention for example.  These results are generally capable of handling continuous or categorical data, and they lead to necessary conditions for Simpson's paradox.  The results are meant mainly for use during sensitivity analysis.  Sensitivity of bivariate correlation coefficients can be assessed even in the absence of covariate data, since $r$ and $R^2$ values are readily estimated using only subject matter knowledge.  Sensitivity of multiple regression coefficients may be assessed in a similar manner, although using subject matter knowledge to estimate partial coefficients may be difficult.  The general formulation of results within Section \ref{resultssec} is meant for application during analysis of large data sets.  With a large number of covariates it may not be computationally feasible to fit every possible model extension, yet computation of partial coefficients along with Proposition \ref{firstresult} can make model-independent estimation possible.  Model-independent estimation is demonstrated in Section \ref{miesec} with dietary health data from the United Nations.  Section \ref{simpsec} shows how an occurrence of Simpson's paradox implies the reversal of a least-squares estimate, but not vice versa.
\subsection{Simpson's paradox}
\label{simpsec}
\begin{definition}
\label{catys}{\it Simpson's paradox} is the designation for a surprising situation that may occur when two populations are compared with respect to the incidence of some attribute: if the populations are separated in parallel into a set of descriptive categories, the population with higher overall incidence may yet exhibit a lower incidence within each such category \citep{Wagner82}.
\end{definition}
For examples of Simpson's paradox see Section 2 or Section 3 of Wagner's article or the examples section of \citet{Julious94}.  Another well known example occurred when the University of California at Berkeley was sued for gender bias.  Overall, female graduate school applicants were being admitted at a lower rate than males, but within most departments (where autonomous decisions were being made) females were being admitted at higher rates than males.  The bias did not reverse within every department, yet the authors still chose to describe the situation as ``a paradox, sometimes referred to as Simpson's" \citep{Bickel75}.  Similar terminology has been used by \citet{Appleton96}.  Common to these examples is a reversal of the purported effect of population on incidence.  We thus propose a weaker definition using the terminology of least-squares regression.
\begin{definition}
\label{reversal}
Let $\mathbf{x}$ indicate population, $\mathbf{y}$ indicate incidence, and $\mathbf{u}=[\mathbf{u}_1\,...\,\mathbf{u}_k]$ indicate category.  We say that a {\it reversal} of the effect of $\mathbf{x}$ on $\mathbf{y}$ has occurred if
\[\mathrm{sign}(\hat{\beta}_{\mathbf{x}|\mathbf{u}}(\mathbf{y}))\neq \mathrm{sign}(\hat{\beta}_{\mathbf{x}}(\mathbf{y})).\]
\end{definition}
\begin{lemma}
If Simpson's Paradox has occurred then a Reversal has occurred.
\end{lemma}
\begin{proof}
There are $k+1$ categories.  Let $j$ be an element of $\{0,1,...,k\}$.  Let $\hat{\beta}_{\mathbf{x}}(j)$ represent the least-squares slope coefficient for $\mathbf{x}$ when $\mathbf{y}$ is regressed onto $\mathbf{x}$ over only those data in category $j$.  For every $j$ we assume that $\hat{\beta}_{\mathbf{x}}(j)>0$.

Given a regression of $\mathbf{y}$ onto $[\mathbf{x}\,\mathbf{e}\,\mathbf{u}_1\,...\,\mathbf{u}_k]$ we have least-squares fitted coefficients $\{\hat{\beta}_{\mathbf{x}},\hat{\beta}_0,\hat{\beta}_1,...\hat{\beta}_k\}$.  The sum of the squares is a function of $(\beta_{\mathbf{x}},\beta_0,\beta_1,...\beta_k)$, and $(\hat{\beta}_{\mathbf{x}},\hat{\beta}_0,\hat{\beta}_1,...\hat{\beta}_k)$ is the minimizer.  For $i\in\{0,1\}$ let $\bar{y}(j,i)$ denote the mean of those observations in category $j$ with $X=i$.

Observe how for every $j$ we have $\hat{\beta}_j<\bar{y}(j,1)$.  Thus we consider only $(\beta_0,\beta_1,...\beta_k)$ such that $\beta_j<\bar{y}(j,1)$ for each $j$.  Note how for any such tuple with $\alpha>0$ that the sum of the squares when $\beta_{\mathbf{x}}=\alpha$ is less than the sum of the squares when $\beta_{\mathbf{x}}=-\alpha$.  Therefore $\hat{\beta}_{\mathbf{x}}>0$.
\end{proof}

The preceding proof shows how the criterion for Simpson's paradox is strictly stronger than the criterion for a reversal.  The corollaries of Section \ref{resultssec} can thus be modified into theorems for Simpson's paradox.   Let $\mathbf{x}$ indicate population, let $\mathbf{y}$ indicate attribute presence, and let $\mathbf{u}$ indicate category.  Let $P$ be a partition of $\{\mathbf{u}_1,...,\mathbf{u}_k\}$ into $q$ cells, where $1\leq q\leq k$.  Let $\mathbf{t}$ be a matrix with columns indicating cell membership.  Note that $R^2$ for $\mathbf{t}$ is less than or equal to $R^2$ for $\mathbf{u}$.  Simpson's paradox is with respect to $\mathbf{t}$.  Coefficients of determination for sets of indicator variables are well defined as long as the same non-zero quantity is used to indicate membership for all individuals within a specific category.  Generally, the coefficient of determination can be defined as a geometric property of linear subspaces, and thus it is invariant under change of basis.
\begin{theorem}[Strong, Necessary Condition for Simpson's paradox]
Simpson's paradox can not occur unless
 \[R(\mathbf{u},\mathbf{x})R(\mathbf{u},\mathbf{y})>|r(\mathbf{x},\mathbf{y})|.\]
\end{theorem}
\begin{theorem}[Weak, Necessary Condition for Simpson's paradox]
Simpson's paradox can not occur unless \[R^2(\mathbf{u},\mathbf{v})>r^*.\]
\end{theorem}

Necessary conditions for reversals of least-squares estimates are necessary conditions for Simpson's paradox, but these conditions are not adequate for other varieties of ecological fallacy (see \citet{Piantadosi88}).  This distinction is relevant throughout the next subsection, where we analyze country-level effects.  Analysis of reversals can be used to determine whether or not these country-level effects are due to categorization into continents, as might be suggested if confounding due to ethnicity or genetic makeup is suspected, but further assumptions would be required in order to pass to continent-level or individual-level results.  Our focus here is not on multilevel analysis nor traditional inference but rather a technique that adjusts for an indeterminate set of covariates.
\subsection{Model-independent estimation}
\label{miesec}
In this subsection we demonstrate methodology with data that was recorded in 2008 and 2009 by the United Nations.  The data was obtained in 2013 from three different sources: the World Health Organization (WHO), the Human Development Report Office (HDRO), and the Food and Agriculture Organization (FAO).  For each of 155 countries, age-adjusted, mean, total cholesterol levels \citep{WHO08} and Human Development Index (HDI) scores \citep{HDRO} were retrieved, along with per capita consumption rates for meat, milk, eggs \citep{FAO}, fish, and animal fats \citep{FAOSTAT}.\symbolfootnote[2]{These data were retrospectively selected for instructive demonstrations of model-independent estimation.  Variables were chosen for pedagogical reasons unrelated to scientific study of cholesterol, and causal conclusions are not intended nor implied by these demonstrations.}  

HDI is an index that measures the state of human development within a country, utilizing indicators relating to life expectancy, educational attainment, and income per capita.  Among the variables just mentioned, HDI correlates most strongly with cholesterol levels, with a correlation coefficient of approximately $r=0.91$.  (Henceforth we round all estimates to the nearest hundredth.)  A bivariate plot of HDI and cholesterol data is shown in Figure \ref{astrongcor}.  Analysis of reversals leads to the belief that such a strong correlation is unlikely to be reversed by controlling for covariates.

\begin{figure}[t]
\centering
\includegraphics[width=4.6in]{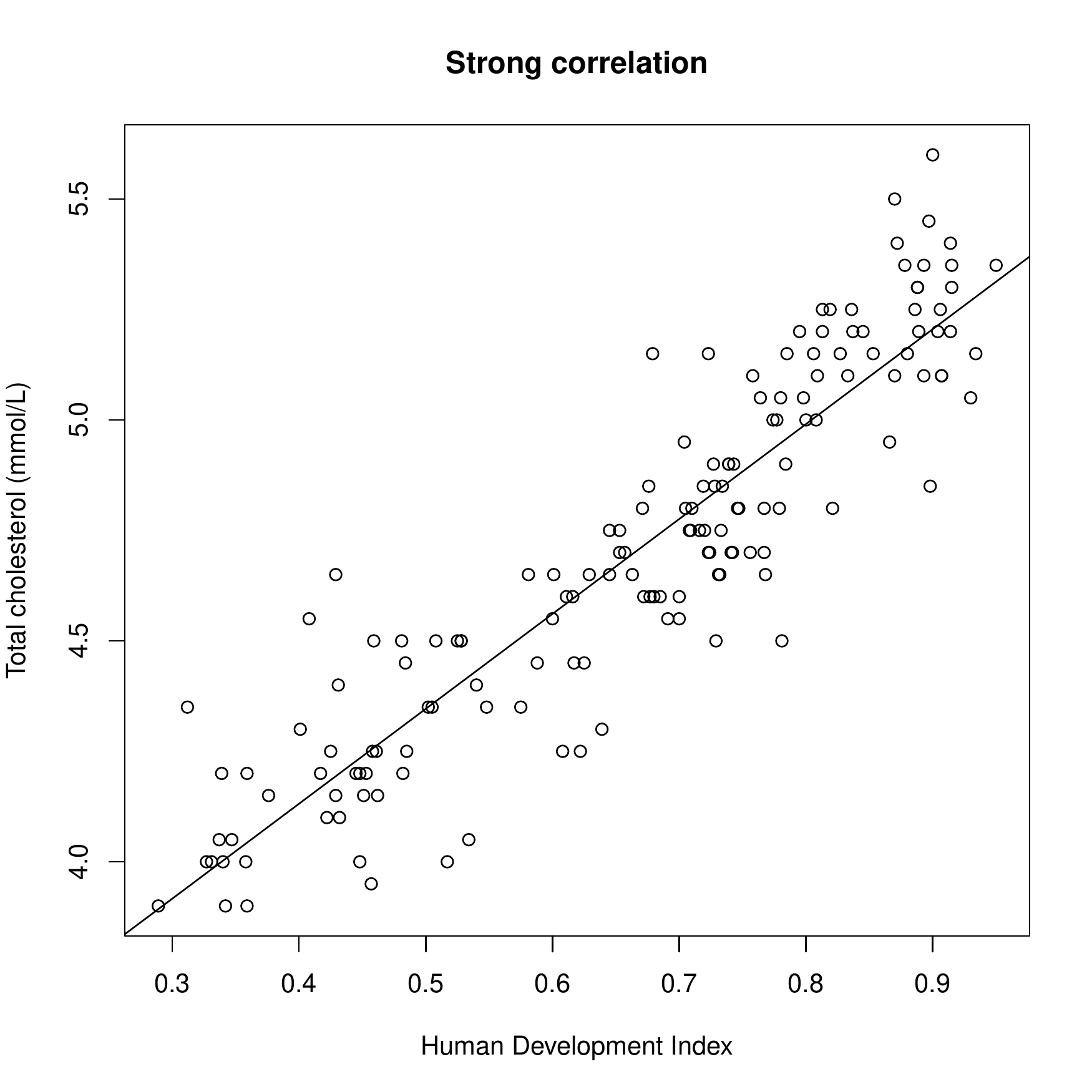}
\caption{A scatter plot (n=155) showing a strong correlation ($r\approx 0.91$) between development and mean total cholesterol levels at the country level}
\label{astrongcor}
\end{figure}

Meat consumption is measured in kg/person/year and includes consumption of pig, poultry, cattle, and sheep.  The observed correlation between meat consumption and cholesterol is $0.81$, while the observed correlation between meat consumption and HDI is $0.82$.  These numbers, while impressive, are not strong enough to induce a reversal.  The magnitude of their product, $0.66$, is less than that required for a reversal, $r=0.91$, and therefore by Corollary \ref{oneway1} we can be sure that the direction of the estimate for the effect of HDI on cholesterol is not sensitive to adjustment by control for meat consumption.

The actual fitted linear model of cholesterol in terms of HDI and meat consumption gives more information.  When fit over standardized data, so as to allow for comparison across differing units, HDI remains the dominant explanatory variable.  Its fitted coefficient is $0.026$, with an associated $t$ statistic of $13.8$ ($\text{p}\approx 10^{-15}$), while the fitted coefficient for meat is $0.006$, with a $t$ statistic of $3.0$ ($\text{p}\approx 0.004$).  The retained importance of HDI is visually evident in Figure \ref{3dplot}.  At nearly all levels of meat consumption the estimate for the effect of increasing HDI on cholesterol remains strongly positive.

\begin{figure}[t]
\includegraphics[width=4.6in]{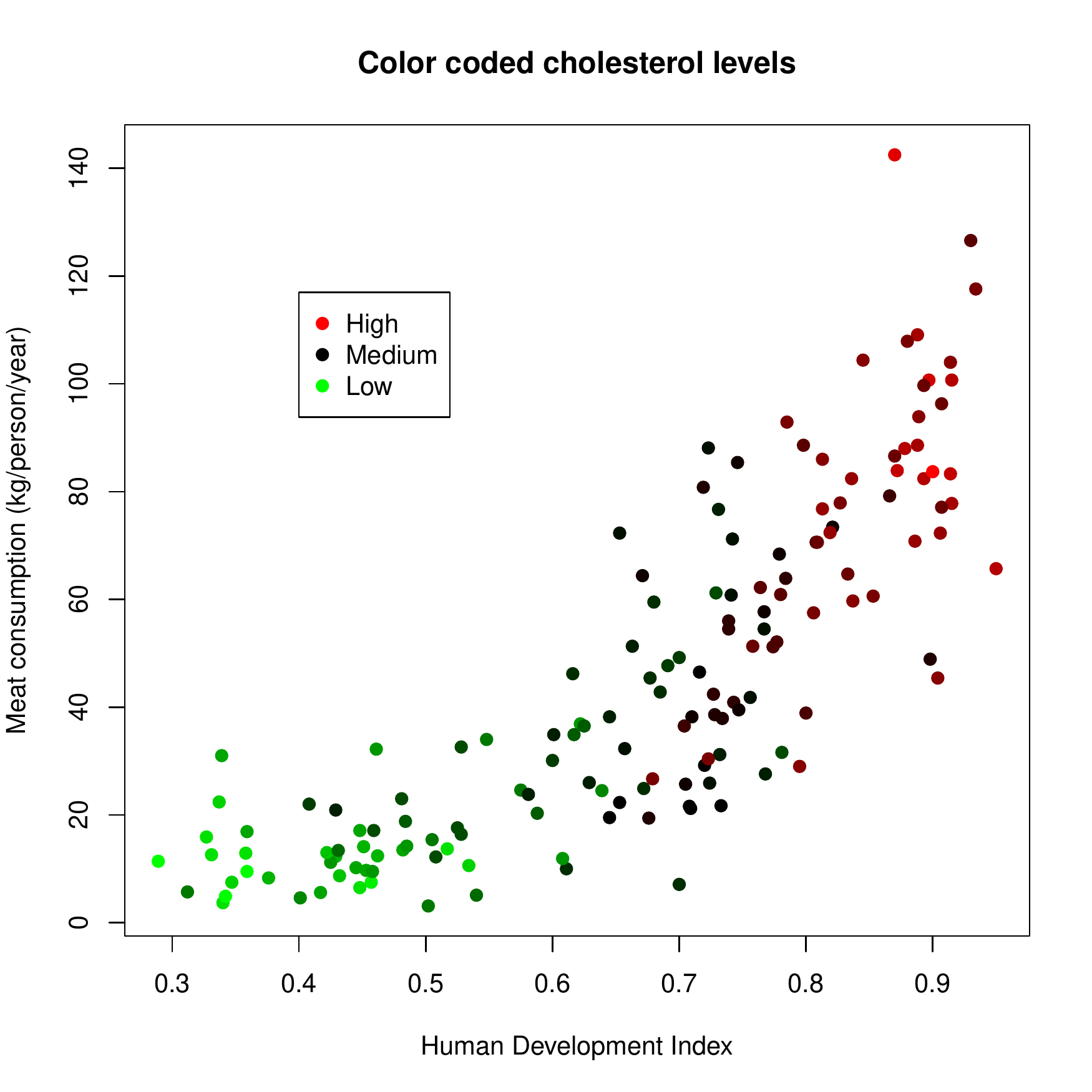}
\caption{A plot (n=155) showing the trivariate relationship between color-coded, mean, total cholesterol levels, development, and meat consumption, at the country level: conditional on meat consumption the relationship between development and cholesterol appears linear and the estimate for the effect of development on cholesterol remains strongly positive; conditional on development the estimate for the effect of meat consumption on cholesterol is much smaller in magnitude.}
\label{3dplot}
\end{figure}

Next we adjust for meat, milk, eggs, fish, and animal fat, simultaneously.  A large model of cholesterol in terms of HDI and all these dietary variables is summarized in Table \ref{Taball}.  HDI remains highly significant ($t=8.49, p\approx10^{-13}$), and its dominance is not unexpected.  We know from Corollary \ref{oneway2} that for a reversal to occur the dietary variables' coefficient of determination for $\mathbf{v}$ (the standardized sum of HDI and cholesterol) must be larger than $r^*=2r/(r+1)=0.96$, and calculation reveals that this coefficient is only $0.84$.  Therefore, adjustment for any subset of the dietary variables can not induce a reversal.  This conclusion could conceivably have been reached even in the absence of data, since coefficients of determination can be estimated with subject matter knowledge.

\begin{table}[t]
\caption{A linear model of cholesterol fit to standardized, country-level data.  HDI is the dominant explanatory variable, even when the five dietary variables are combined into one variable, namely the vector of fitted values from the dietary model of Table \ref{Tabdiet}. }
\label{Taball}
\centering
    \begin{tabular}{lccc}
\toprule
explanatory variable    & fitted slope coefficient & t statistic & two-sided p value \\
\midrule
HDI&0.58& 8.49 & $\approx10^{-13}$   \\
meat &0.11& 1.95 & 0.05  \\
milk 	&0.08& 1.50 & 0.14   \\
eggs &0.12& 2.54 & 0.01   \\
fish &0.07& 2.13 & 0.03   \\
animal fat &0.10 & 2.34 & 0.02   \\
 \bottomrule
\end{tabular}
\end{table}

It is more difficult to estimate a partial coefficient of determination with subject matter knowledge.  Suppose that subject matter knowledge has lead to a dietary model of cholesterol in terms of only meat, milk, eggs, fish, and animal fat.  This model is summarized in Table \ref{Tabdiet}.  Note the final column, where we have included absolute values of partial correlation coefficients.  These coefficients are computed as partial correlations between a given row's variable and cholesterol, given the remaining dietary data.  Calculation with residual vectors reveals, using either Corollary \ref{oneway1} or Corollary \ref{oneway2}, that HDI is not capable of inducing any reversals.  With $k$ covariates similar calculation would be done for the whole set of covariates at once.

\begin{table}[t]
\caption{A linear model of cholesterol that has not been adjusted for HDI.  Partial correlations have been computed between a given row's variable and cholesterol, given the remaining dietary variables.  Higher partial correlations indicate stability.}
\label{Tabdiet}
\centering
    \begin{tabular}{lcccc}
\toprule
variable & slope coefficient & t statistic & p value & partial correlation\\
\midrule
meat &0.35& 5.93 & $\approx10^{-7}$ & 0.44  \\
milk 	&0.22& 3.70 & 0.0003 & 0.29 \\
eggs &0.32& 6.19 & $\approx10^{-8}$ & 0.28 \\
fish &0.14& 3.59 & 0.0005 & 0.45 \\
animal fat & 0.09  & 1.77 & 0.0786 & 0.14 \\
 \bottomrule
\end{tabular}
\end{table}

\section{Discussion}
\label{dsec}
Proposition \ref{firstresult} and its corollaries have been designed for use during analysis of large data sets, especially when the goal is to estimate the direction of a causal effect of $X$ on $Y$ by adjusting for covariates.  Suppose a model of $\mathbf{y}$ has been fit to $\mathbf{x}$ and $\mathbf{w}$, and confounding by some indeterminate subset $\mathbf{s}\subseteq \mathbf{u}$ is suspected.  There are $2^k$ subsets to consider, each associated with a particular model extension, and it may not be feasible to fit all possible models.  However, by fitting a single model of $\mathbf{v}$ in terms of $\mathbf{u}$, if the $R^2$ value is small compared to $r^*$, then the technique of model-independent estimation can be implemented.  That is the content of Corollary \ref{oneway2}, and Corollary \ref{oneway1} is similar.

Related theory exists within the field of econometrics.  Here we have dealt with model extensions, while econometricians have already dealt with model contractions.  They have studied reversals by assuming a larger model and an effect of interest, along with conditions on a set of variables to be removed.  Using $t$ and $F$ statistics, \citet{Leamer75} showed how reversals can only occur if the set of variables to be dropped is more significant than the variable of interest.  \citet{Visco78} showed that this condition is not sufficient, and he also derived necessary and sufficient conditions for a reversal when only a single variable is dropped.  \citet{Oksanen87} rephrased the conditions using partial correlation.  \citet{McAleer86} and \citet{Giles89} presented generalizations.  However, using the words of \citet[p 126]{Imbens03} ``One is not interested in what would have happened in the absence of covariates actually observed, but in biases that are the result from not observing all relevant covariates.''

For example, consider smoking and lung cancer.  A simple causal graph is inadequate, because of complicated relationships between smoking, lung cancer, and confounding variables \citep[p 424]{Pearl09a}.  For instance, the US Environmental Protection Agency (EPA) lists (indoor exposure to) radon (gas) as the second leading cause of lung cancer in the United States \citep{Radonb}, and there is evidence of interaction between radon gas and smoke \citep[Appendix C, p 239]{Radona} .  It wasn't a perfect model, but rather an inequality \citep[Appendix A]{Cornfield59} that played a critical role in allowing the US Surgeon General to conclude that cigarette smoking is causally related to lung cancer in man \citep{Lin98}.   In response to Fisher's constitution hypothesis \citep{Fisher58}, Cornfield et al. stated that ``the magnitude of the excess lung-cancer risk among cigarette smokers is so great that the results can not be interpreted as arising from an indirect association of cigarette smoking with some other agent or characteristic, since this hypothetical agent would have to be at least as strongly associated with lung cancer as cigarette use; no such agent has been found or suggested.''

A limitation of reversal analysis is its emphasis on direction rather than magnitude.  There is much literature dealing more exactly with omitted variable bias.  It can be specified as a complicated matrix expression \cite[Chaper 3]{Seber03}.  It can be factored into a ratio of standard errors, an F statistic, and a partial coefficient of determination \citep{Hosman10}.  Expressions bounding the t values of the larger model can be written in terms of coefficients of determination, under certain assumptions \citep{Frank00}.  Assuming binary treatment, sensitivity can be assessed with distributional assumptions for the confounding variables, along with knowledge of how they affect the response \citep{Lin98}.  See also \citet{Rosenbaum83b}.  In general, more exact results require more detailed assumptions.  There are few assumptions underlying the analysis of reversals.  Precision has been traded for the possibility of model-independent estimation.

Analysis of reversals has produced necessary conditions for Simpson's paradox, revealed geometric symmetry within the column space of data sets, and lead to the possibility of model-independent estimation---a technique for identifying effects that are invariant across a class of models.  To determine the direction of an effect, either Corollary \ref{oneway1} or Corollary \ref{oneway2} can be applied, and only basic knowledge of $r$ and $R^2$ is required.  Note that $r$ alone is not sufficient.  Table  \ref{tab2} gives an example where $\mathbf{u}_1$ and $\mathbf{u}_2$ both correlate arbitrarily weakly with both $\mathbf{x}$ and $\mathbf{y}$, yet $[\mathbf{u}_1\,\mathbf{u}_2]$ induces a reversal.  Also, partial coefficients are required.  Table \ref{tab1} gives a related example where a single vector $\mathbf{u}$ is not correlated with $\mathbf{x}$ nor $\mathbf{y}$, yet it induces a reversal nonetheless, by activating a previously dormant $\mathbf{w}$.  Finally, even with $\mathbf{w}=\emptyset$ it is not possible to conduct model-independent estimation while retaining $r(\widehat{\mathbf{x}_{|\mathbf{w}}}(\mathbf{u}),\widehat{\mathbf{y}_{|\mathbf{w}}}(\mathbf{u}))$ in its entirety for possibly stronger logical reasoning.  Table \ref{tab3} gives an example where such theory would suggest (correctly) that a reversal is not possible due to $[\mathbf{u}_1\,\mathbf{u}_2]$, but both $\mathbf{u}_1$ and $\mathbf{u}_2$ individually lead to reversals.

\begin{table}[h]
\caption{A counterexample showing the need for $R^2$: $\hat{\beta}_{\mathbf{x}}(\mathbf{y})=r(\mathbf{x},\mathbf{y})\approx0.5$, and as $\epsilon \downarrow 0$, $r(\mathbf{u}_1,\mathbf{x})=r(\mathbf{u}_1,\mathbf{y})\downarrow 0$ and $r(\mathbf{u}_2,\mathbf{x})=r(\mathbf{u}_2,\mathbf{y})\downarrow 0$, while $R(\mathbf{u},\mathbf{y})R(\mathbf{u},\mathbf{y})\approx0.75$ and $\hat{\beta}_{\mathbf{x}|\mathbf{u}}(\mathbf{y})=-1$.}
\label{tab2}
\centering
    \begin{tabular}{cccc}
\toprule
        $\mathbf{y}$ & $\mathbf{x}$ & $\mathbf{u}_1$ & $\mathbf{u}_2$\\
\midrule
$(\sqrt{2}+3)/2$	& $(-\sqrt{2}+3)/2$ & $\epsilon/\sqrt{2}$ & $\epsilon/\sqrt{2}$ \\
$(\sqrt{2}-3)/2$	& $(-\sqrt{2}-3)/2$  & $-\epsilon/\sqrt{2}$ & $-\epsilon/\sqrt{2}$ \\
$-1/2$	                     & $1/2$                               & $1$ & $-1$ \\
$-1/2$	                     & $1/2$                                 & $-1$ & $1$ \\
 \bottomrule
\end{tabular}
\end{table}
\begin{table}[h]
\caption{A counterexample showing the need for partial coefficients: as $\delta\downarrow 0$, $\hat{\beta}_{\mathbf{x}|\mathbf{w}}(\mathbf{y})\approx0.5$, $r(\mathbf{u},\mathbf{x})=r(\mathbf{u},\mathbf{y})=0$, yet $\hat{\beta}_{\mathbf{x}|\mathbf{w},\mathbf{u}}(\mathbf{y})\approx-0.4$, while $r(\mathbf{w},\mathbf{x})=r(\mathbf{w},\mathbf{y})\downarrow 0$.}
\label{tab1}
\centering
    \begin{tabular}{cccc}
\toprule
        $\mathbf{y}$ & $\mathbf{x}$ & $\mathbf{w}$ & $\mathbf{u}$\\
\midrule
$(\sqrt{2}+3)/2$	& $(-\sqrt{2}+3)/2$ & $\delta/\sqrt{2}$ & $0$ \\
$(\sqrt{2}-3)/2$	& $(-\sqrt{2}-3)/2$  & $-\delta/\sqrt{2}$ & $0$ \\
$-1/2$	                     & $1/2$                               & $1$ & $-1$ \\
$-1/2$	                     & $1/2$                                 & $-1$ & $1$ \\
 \bottomrule
\end{tabular}
\end{table}
\begin{table}[h]
\caption{A counterexample showing how model-independent estimation is not possible with full use of $r(\widehat{\mathbf{x}_{|\mathbf{w}}}(\mathbf{u}),\widehat{\mathbf{y}_{|\mathbf{w}}}(\mathbf{u}))$ and Proposition \ref{firstresult}: $\hat{\beta}_{\mathbf{x}}(\mathbf{y})\approx0.5$, and for small, positive $\epsilon$ and $\delta$, as $(\epsilon,\delta)\to (0,0)$, $\hat{\beta}_{\mathbf{x}|\mathbf{u}_1,\mathbf{u}_2}(\mathbf{y})\to 1.0$, $\hat{\beta}_{\mathbf{x}|\mathbf{u}_1}(\mathbf{y})\to -1.0$, and $\hat{\beta}_{\mathbf{x}|\mathbf{u}_2}(\mathbf{y})\to -1.0$.}
\label{tab3}
\centering
    \begin{tabular}{cccc}
\toprule
        $\mathbf{y}$ & $\mathbf{x}$ & $\mathbf{u}_1$ & $\mathbf{u}_2$\\
\midrule
$(\sqrt{2}+3)/2$	& $(-\sqrt{2}+3)/2$ &$(\epsilon+3\sqrt{2})/2$ & $(-\epsilon+3\sqrt{2})/2$ \\
$(\sqrt{2}-3)/2$	& $(-\sqrt{2}-3)/2$  & $(\epsilon-3\sqrt{2})/2$ & $(-\epsilon-3\sqrt{2})/2$ \\
$-1/2$	                     & $1/2$                               & $(-\epsilon+\delta\sqrt{2})/2$ & $(\epsilon+\delta\sqrt{2})/2$ \\
$-1/2$	                     & $1/2$                                 & $(-\epsilon-\delta\sqrt{2})/2$ & $(\epsilon-\delta\sqrt{2})/2$ \\
 \bottomrule
\end{tabular}
\end{table}
\newpage

\bibliographystyle{Chicago}

\bibliography{bibliography}

\end{document}